\documentclass[12pt]{amsart}

\pdfoutput=1

\usepackage[text={420pt,660pt},centering]{geometry}

\usepackage{color}
\usepackage{esint,amssymb}
\usepackage{graphicx}
\usepackage{MnSymbol}
\usepackage{mathtools}
\usepackage[colorlinks=true, pdfstartview=FitV, linkcolor=blue, citecolor=blue, urlcolor=blue,pagebackref=false]{hyperref}
\usepackage{microtype}

\usepackage{bm}
\usepackage{scalerel} 
\usepackage{dsfont}
\usepackage[font={footnotesize}]{caption}

\usepackage[notcite,notref,color]{showkeys}

\usepackage{lipsum}
\definecolor{darkgreen}{rgb}{0,0.5,0}
\definecolor{darkblue}{rgb}{0,0,0.7}
\definecolor{darkred}{rgb}{0.9,0.1,0.1}

\newtheorem{proposition}{Proposition}
\newtheorem{theorem}[proposition]{Theorem}

\theoremstyle{definition}

\newtheorem{problem}[proposition]{Problem}

\newcommand{\cref}[1]{Corollary~\ref{c.#1}}

\numberwithin{equation}{section}
\numberwithin{proposition}{section}

\begin{document}

\title[The ACA and the IRS Fixed Point Iteration]{The Affordable Care Act and the IRS Iterative Fixed Point Procedure}

\begin{abstract}
We model the quantities appearing in Internal Revenue Service (IRS) tax guidance for calculating the health insurance premium tax credit created by the Patient Protection and Affordable Care Act, also called Obamacare. We ask the question of whether there is a procedure, computable by hand, which can calculate the appropriate premium tax credit for any household with self-employment income. We give an example showing that IRS tax guidance, which has had self-employed taxpayers use an iterative fixed point procedure to calculate their premium tax credits since 2014, can lead to a divergent sequence of iterates.
As a consequence, IRS guidance does not calculate appropriate premium tax credits for tax returns in certain income intervals, adversely affecting eligible beneficiaries.
A bisection procedure for calculating premium tax credits is proposed. We prove that this procedure calculates appropriate premium tax credits for a model of simple tax returns; and apparently, this procedure has already been used to prepare accepted tax returns.
We outline the problem of finding a procedure which calculates appropriate premium tax credits for models of general tax returns. While the bisection procedure will work with the tax code in its current configuration, it could fail, in states which have not expanded Medicaid, if a certain deduction were to revert to an earlier form. Future policy objectives might also lead to further problems.
%
\end{abstract}

\author[S. J. Ferguson]{Samuel J. Ferguson}
\address[S. J. Ferguson]{Courant Institute of Mathematical Sciences, New York University, 251 Mercer St., New York, NY 10012}
\email{ferguson@cims.nyu.edu}

\maketitle

\section{Introduction}

In January of 2018, the author took an Uber ride. The author's Uber driver was eligible for a tax credit under the Affordable Care Act, but tax software and government calculators couldn't correctly calculate it. He was going to receive \$0 to help him pay for health insurance, instead of the roughly \$3,000 that the law prescribed. He asked the author to look into the matter, and this led to a mathematical odyssey captured by \emph{Time}'s film crew and a senior writer at \emph{Money} in \cite{TIM}. The author is motivated to write this article on his research by the communication \cite{OCC}, that the IRS will include reference to it in its guidance after publication in a peer-reviewed journal. Then, tax software companies will be able to implement improved procedures without legal liability, relieving the current computational issues affecting beneficiaries of the Affordable Care Act.

\section{The Affordable Care Act's premium tax credit}

The Patient Protection and Affordable Care Act 
\cite{ACA},
%
also called Obamacare, is a federal law of the United States of America, passed in 2010. Among other things, the law makes qualified health insurance ``affordable" for every American household with ``modified income" $M$ in the range (we do not give definitions of ``affordable" or ``modified income" now, so the reader is asked for patience)
\[
F\leq M\leq 4F.
\]
Here, $F$ is the \emph{federal poverty line}, a governmentally-prescribed number\footnote{For the sake of intuition, in the continental United States in 2018, for a household with $n$ people, $F$ is approximately $\text{\$8,000}+n\cdot\text{\$4,000}$. Thus, a household of one person typically has a \emph{federal poverty line} of about \$12,000, and a household of four people typically has a federal poverty line of about \$24,000.} depending on household size, state, and tax year, which adjusts annually according to a specified notion of inflation.

To convert from units of dollars to mathematical quantities without units, we define $m$ by the equation
\[
m=M/F.
\]
Now, how does the law make health insurance affordable when $1\leq m\leq 4$? It does so by creating the following tax device.

\emph{Step 1.} Given $m$ in the interval $[1,4]$, a governmentally-prescribed percentage is specified, depending on $m$. This is the percentage of $M$ that the government expects the household to be able to affordably pay for health insurance. For example, it is typically less than $0.03$ for households near $m=1$, and close to $0.09$ for households near $m=4$, though the function is different each year. We denote the value of this percentage function as $\%(m)$, and $\%(m)$ is in the interval $(0, 0.1)$ for all $m$ in $[1,4]$ so far, though future years could call for higher percentages. This function is monotone increasing and, so far, discontinuous each year. Fortunately, by the grace of Congress, it possesses right continuity on $[1,4]$, independent of the tax year.

\emph{Step 2.} The percentage $\%=\%(m)$ is multiplied by income $M$, and this number, $\%\cdot M$, is the household's \emph{expected contribution} towards health insurance. This is then compared with the (unsubsidized) sum of the costs of benchmark annual health insurance premiums for the household members. Each household member's cost depends on county of residence, age, and smoking habits. The government ``picks up the tab," that is, the government is willing to pay all of the cost of the benchmark health insurance premium which is not covered by the expected contribution. It does so by means of a (refundable) tax credit that can be obtained in advance at the time that premiums must be paid. In fact, the money may be sent directly from the government to the insurance company, so that the household may not even be aware of the (unsubsidized) total cost of the insurance.

\emph{Simplified Example.} Say an unmarried 60-year-old nonsmoking male forms a household of $1$ person in Dutchess County, in the state of New York, in 2018. Suppose $P$ is $\$500$ per month, $F=\$12,000$ per year, $M=\$48,000$ per year in 2018, and $\%(4)=0.09$. Then, his expected contribution is
\[
\%\left(\frac{M}{F}\right)\cdot M = 0.09\cdot \text{\$48,000} = \text{\$4,320},
\]
and the annual benchmark cost is
\[
12\cdot \text{\$500} = \text{\$6,000}.
\]
If he buys the benchmark insurance, then the government pays the rest,
\[
\text{\$6,000}-\text{\$4,320}=\text{\$1,680}.
\]
Thus, \$1,680 is the amount of Obamacare's health insurance premium tax credit for the household.

The full cost of the benchmark insurance premiums is paid with the premium tax credit and the expected contribution, so the household doesn't have to spend more than its expected contribution for this\footnote{It should be noted that a household does not need to buy the benchmark insurance, and can buy other qualified insurance. The government is still willing to pay \$1,680 or the full cost of the chosen insurance, whichever is less (the government cannot pay more than the full cost of insurance). For simplicity, however, we will assume that the benchmark insurance, which is ``the second lowest cost silver plan" on the government exchange for the household's county of residence, is the insurance plan actually purchased.} insurance, which the government, by definition, considers to be affordable for the household. Thus Obamacare, according to government definitions, makes qualified health insurance affordable for all American households with modified incomes $M$ satisfying $F\leq M\leq 4F$.

\section{The Problem}

The reader will have noticed that we have yet to define $M$. There's a reason for that. For households with sufficient income from \emph{self-employment} (an independent contractor, a private tutor, and a driver associated with a ridesharing app are all likely to be considered self-employed), the value of $M$ can be tricky to find. Yet it seems worthwhile to figure this out for self-employed households. In 2014, self-employed workers were ``almost three times more likely'' than other workers obtain health insurance from the government exchanges created by Obamacare, according to
\cite{LMT},
%
so self-employed households form a sizable proportion of beneficiary households. Thus, we are motivated to address any computational issues that they may face; such issues could potentially affect a large number of people.

Say, for simplicity, that all of a household's income comes from a single business generating an ``earned income" from self-employment of $I$. If the household's health insurance is all purchased through this single business (so $I$ is at least as large as the annual insurance cost), then some amount $D$ of that cost can be deducted from $I$ (meaning no federal taxes are paid on the amount deducted). If there are no other sources of so-called ``above-the-line" deductions besides $D$, then the modified income $M$ is then given by the equation
\[
M=I-D.
\]
We could say that the household has a ``simple" tax return in this case, since it only has one source of income and one above-the-line deduction.

We now have two constraints: first,
\[
0\leq D\leq Q,
\]
where $Q$ is the full (unsubsidized) cost of the insurance, since the government doesn't permit more to be deducted than was possible to pay; second,
\[
D+C\leq Q,
\]
where $C$ is the amount of Obamacare's premium tax credit for the household with modified income $M=I-D$. This second constraint comes from the fact that the government doesn't permit more total dollars in deductions and tax credits than was possible to pay. Without this restriction, it might be possible for an enterprising individual to obtain health insurance at a negative cost to themselves, presumably contrary to the wishes of taxpayers.

\begin{problem}
What is a procedure, computable by hand in a reasonable time, that finds the appropriate health insurance deduction $D$ and premium tax credit $C$ for any self-employed household? The federal poverty line $F>0$, annual cost $Q$ of qualified health insurance premiums (equal to the benchmark cost), earned self-employment income $I\geq Q$, and monotone increasing, right continuous percentage function $\%:[1,4]\to (0,0.1)$ are given.
\end{problem}

What does ``appropriate" mean, above? The appropriate choice of $D$ and $C$ is the choice which maximizes the tax benefit for the beneficiary. That is, we wish to find $D$ and $C$ for which the solution of the maximization problem
\[
\max_{0\leq D\leq Q\ :\ C=C(D),\ D+C\leq Q}\left(C+\overline{t}D\right)
\]
is attained. Here, $C(D)=Q-\%\left(\tfrac{I-D}{F}\right)\cdot \left(I-D\right)$ if $1\leq\tfrac{I-D}{F}\leq 4$ and $C(D)=0$ if $\tfrac{I-D}{F}>4$, as no assistance is given to households with modified incomes beyond the $4F$ threshold. The quantity $\overline{t}$ is the effective income tax rate that would be paid on the amount $D$ of income if it were not deducted, so
\[
\overline{t}=\frac{T(I)-T(I-D)}{I-\left(I-D\right)},
\]
where, for each year and household, $T$ is a monotone-increasing function which assigns, to a given amount, the federal income tax liability for that income.

What does ``computable by hand in a reasonable time" mean, above? It means that the Internal Revenue Service (IRS) could put it into its tax guidance. Informally, that would mean that the IRS does not consider it overly onerous to require of any American of sound mind, even if removed from the conveniences of modern technology. So, for example, if we try all possible values (rounded to the nearest dollar) for $D$ and $C$ that satisfy the constraints, some pair of values will yield the maximum tax benefit, and so will give appropriate $D$ and $C$. But the IRS would likely consider having to try all possible pairs of constrained values $(D, C)$ to be an overly burdensome computational task for an American unable to access a computer, smartphone, or calculator. Thus, this solution, although guaranteed to succeed mathematically, is not a procedure that any household can ``compute by hand in a reasonable time," so it is not a proposed answer to the question. Likewise, the constrained maximization problem that leads to appropriate $D$ and $P$ can be converted, given sufficient knowledge of the function $\%(m)$, to some algebraic equation in $D$ for each taxpayer, since $\%(m)$ is a piecewise-polynomial function and the constraints are easily-analyzed inequalities. However, it would not be reasonable to require an American who has never learned algebra, and is removed from technology, to spend sufficient time to discover the necessary algebra for the solution of the equation (and the associated numerical computation of, say, square roots) independently. So the task is to create an algorithm or procedure which can be implemented in (not too many) steps that just involve addition, subtraction, multiplication, and division. Also, it is desirable to avoid using too many specific details about the discontinuous function $\%(m)$, as it changes each year, and it would be best to have a consistent procedure, independent of the tax year.

If we can solve the above problem, then it is natural to generalize the problem to arbitrary tax returns, which are not necessarily simple. We can then check whether our solution for simple tax returns, appropriately generalized, will work for all tax returns. We turn now to current IRS guidance for self-employed taxpayers who qualify for both a health insurance deduction $D$ and a premium tax credit $C$. This guidance might be viewed as an attempted solution of this question.

\section{The IRS Iterative Fixed Point Procedure}

Current IRS guidance asks self-employed Obamacare beneficiaries to calculate their tax credit and deduction as the rounded $x$- and $y$-coordinates of the limit of
\[
(C_{1}, D_{1}), (C_{2}, D_{2}), (C_{3}, D_{3}),\dots,
\]
points in $\mathbb{R}^{2}$ obtained via a fixed-point iteration. The iteration is defined by
\[
(C_{n+1}, D_{n+1})=(C(D_{n}),Q-C(D_{n}))
\]
for integers $n\geq 1$, starting from an initial condition $(C_{1}, D_{1})$ given by $C_{1}=\$0$ and $D_{1}=Q$. The sequence \emph{converges in the IRS sense} if and only if, when rounding to the nearest penny after each intermediate calculation, there exists a positive integer $N$ such that
\[
\|(C_{M}, D_{M})-(C_{N}, D_{N})\|_{\infty}<\varepsilon_{0}
\]
for all integers $M>N$, with $\varepsilon_{0}=\$1$, where $\|(x,y)\|_{\infty}=\max(|x|,|y|)$.
When the sequence fails to converge in the IRS sense, IRS guidance suggests that beneficiaries accept $D_{2}$ as their health insurance deduction and $C_{3}$ as their premium tax credit, in what IRS guidance calls the ``simplified procedure." The best of the available tax software appears to extend the simplified procedure, allowing taxpayers to take at most
\[
D_{0}:=\liminf_{n\to\infty}D_{n}
\]
as their deduction, and $C(D_{0})$ as their premium tax credit. In the case that we don't have convergence in the IRS sense, however, that value for $C$ is always too small, and often \$0. Such inappropriate values are what tax software and government calculators give now. Unsurprisingly, the IRS says that self-employed taxpayers ``may have difficulty'' computing the tax credit, according to
\cite{IRS},
%
the IRS document which introduced these procedures. However, IRS guidance actually says that \emph{any} legal method may be used to find the right amount, although it gives no method which always works to compute appropriate values, as our examples will show. This leaves room for a new procedure.

\emph{Example.} Say we are considering the 2018 tax year. Then the function $\%(m)$ is defined by
\[
\%(m)=\begin{cases}
j, & 1\leq m < 1.33,\\
k+(\ell-k)\tfrac{m-1.33}{1.5-1.33}, & 1.33\leq m<1.5,\\
\ell+(a-\ell)\tfrac{m-1.5}{2-1.5}, & 1.5\leq m< 2,\\
a+(b-a)\tfrac{m-2}{2.5-2}, & 2\leq m< 2.5,\\
b+(c-b)\tfrac{m-2.5}{3-2.5}, & 2.5\leq m < 3,\\
c, & 3\leq m\leq 4,
\end{cases}
\]
where
\[
(j,k,\ell, a, b, c)=(0.0201, 0.0302, 0.0403, 0.0634, 0.0810, 0.0956).
\]
In particular, we can use this to calculate
\[
C(d)=\begin{cases}
Q-\%(\frac{I-d}{F})\cdot (I-d), & 1\leq\frac{I-d}{F}\leq 4,\\
0, & \frac{I-d}{F}>4.
\end{cases}
\]
Say we have a household in Brooklyn, New York, consisting of one parent and one dependent child between 26 and 29 years of age. The household's relevant federal poverty line is $F=\text{\$16,240}$. Looking up prices for the county, Kings County, we find that the (unsubsidized) cost of benchmark health insurance premiums for the household is \$808.07 per month or, rounding to the nearest dollar, $Q=\text{\$9,697}$ annually. Suppose that the household, altogether, has earned self-employment income from a single business which amounts to $I=\text{\$71,150}$. Following the IRS iterative fixed point procedure, and rounding to the nearest dollar in intermediate steps for simplicity, we obtain
\[
(C_{1}, D_{1})=(\$0, \text{\$9,697})
\]
and
\[
C_{2}=\text{\$9,697}-0.0956\cdot\text{\$61,453},
\]
as \text{\$71,150}-\text{\$9,697}=\text{\$61,453}.
Hence, $C_{2}$ is about \$3,822. Thus,
\[
D_{2}=\text{\$9,697}-\text{\$3,822}=\text{\$5,865}.
\]
In turn, this makes $I-D_{2}=\text{\$65,285}>4F=\text{\$64,960}$, so by our above formula for $C(d)$, we get
\[
C_{3}=0.
\]
Unfortunately, this yields
\[
D_{3}=\text{\$9,697},
\]
putting us back where we started. Hence, the sequence doesn't converge in the IRS sense. On the other hand, if we follow IRS simplified procedures, or the extension found in tax software, we arrive at a premium tax credit of \$0, and a deduction of $D_{2}=\$\text{5,865}$. This is even worse than choosing not to take the premium tax credit at all, and letting $D=\$\text{9,697}$. It turns out that the $\$0$ value for the premium tax credit is not appropriate. If we narrow our search for the deduction $D$ somewhat at the beginning, and perform bisection, then we can do better.

\section{The Proposed Bisection Procedure}

We now propose a bisection procedure, and we prove that it always gives the appropriate Obamacare premium tax credit. The proof will work because, although there may, in general, be multiple discontinuities in the underlying structures that affect potential computations, we do still have left continuity in the function $C=C(D)$, before rounding (and we will neglect rounding in stating and proving our theorem). In other words, we can adapt the proof of the Intermediate Value Theorem which uses bisection, in order to calculate the unknown quantities $D$ and $C$ in a stable way. Note that the following theorem does place restrictions on certain parameters, but the solutions when parameters fall outside of these ranges may be given by simple formulas.

\begin{theorem}
Suppose $F>0$, $Q\geq 0$, $I\geq Q$ are given real numbers, and $\%:[1,4]\to (0, 0.1)$ is a monotone increasing right continuous function. Define the function $C=C(d)$ on $[0,\min(Q,I-F)]$ by letting $C(d)=Q-\%(\frac{I-d}{F})\cdot (I-d)$ if $1\leq \frac{I-d}{F}\leq 4$ and $C(d)=0$ if $\frac{I-d}{F}>4$. If $I$ is in the interval $[F, 4F+Q]$, then let
\[
A_{1}=\max(0,I-4F)
\]
and
\[
B_{1}=\min(Q,I-F).
\]
Suppose also that $Q>\%(4)\cdot 4F$ and $I$ is in $(F+Q-\%(1)\cdot F,4F+Q-\%(4)\cdot 4F]$, so that $A_{1}+C(A_{1})\leq Q$ and $B_{1}+C(B_{1})>Q$. For each positive integer $n$, having obtained $A_{n}$ and $B_{n}$, consider their midpoint
\[
E_{n}=\frac{A_{n}+B_{n}}{2}.
\]
If $E_{n}+C(E_{n})\leq Q$, then let $A_{n+1}=E_{n}$ and $B_{n+1}=B_{n}$. Otherwise, if $E_{n}+C(E_{n})>Q$, then let $A_{n+1}=A_{n}$ and $B_{n+1}=E_{n}$. The sequence $\{A_{n}\}_{n=1}^{\infty}$, defined by this procedure, converges. Moreover, if we define
\[
D=\lim_{n\to\infty}A_{n}
\]
and define $C=C(D)$, then no other pair of quantities, satisfying the constraints, achieves a greater value for the maximization problem
\[
\max_{0\leq D\leq Q\ :\ C=C(D),\ D+C\leq Q}\left(C+\overline{t}D\right)
\]
than $(D, C)$, so it solves the problem.
\end{theorem}

\begin{proof}
First, as the function $C(d)$ is increasing in $d$, and $\overline{t}d=T(I)-T(I-d)\geq 0$ is also an increasing function of $d$, the constrained maximization problem is solved for the largest value of $d$ for which all of the constraints are satisfied. Since $\%(m)$ is right continuous on $[1,4]$, it follows that $C(d)$ is left continuous on $[A_{1}, B_{1}]$. As $\{A_{n}\}_{n=1}^{\infty}$ is increasing, $\{B_{n}\}_{n=1}^{\infty}$ is decreasing, and $A_{n}\leq B_{n}$ for all $n\geq 1$, with $B_{n}-A_{n}=\frac{B_{1}-A_{1}}{2^{n-1}}$, it follows that $\{A_{n}\}_{n=1}^{\infty}$ is a Cauchy sequence. Thus, $\lim_{n\to\infty}A_{n}$ exists, so we can define $D=\lim_{n\to\infty}A_{n}$ and $C=C(D)$. By left continuity of $C(d)$, since we have
\[
A_{n}+C(A_{n})\leq Q
\]
for all $n\geq 1$, it follows that
\[
D+C\leq Q.
\]
This proves that the constraints are satisfied for $(D,C)$. As $D+C\leq Q$ and $B_{1}+C(B_{1})>Q$, it follows in particular that $C<B_{1}$. Given any positive number $\varepsilon$ in $(0,B_{1}-D]$, no matter how small, there exists some positive integer $n$ such that $B_{n}-A_{n}=\frac{B_{1}-A_{1}}{2^{n-1}}<\varepsilon$. Then, as $D$ lies in $[A_{n}, B_{n}]$, we have $D+\varepsilon\geq B_{n}$, so
\[
(D+\varepsilon)+C(D+\varepsilon)\geq B_{n}+C(B_{n})>Q.
\]
Hence, $D+\varepsilon$ and $C(D+\varepsilon)$ do not satisfy the constraints. As $\varepsilon>0$ was arbitrary, it follows that no pair $(d,c)$ with $d>D$ can satisfy the constraints. As the quantity to be maximized, $C(d)+\overline{t}d=C(d)+\left(T(I)-T(I-d)\right)$, is increasing in $d$, and $D$ is the largest value of $d$ which satisfies the constraints, it follows that the pair $(D,C)$ solves the maximization problem.
\end{proof}

\emph{Example.} If we perform the bisection procedure on the example from the preceding section, we find that the household receives \$3,489 as a premium tax credit, and \$6,208 in deductions, a far better outcome than what was obtained previously: every extra dollar in refundable tax credits could be could be worth 3, 4, 5, or more dollars in deductions, depending on the details of the household's tax profile. That's because reducing taxable income by a dollar only reduces taxes by the fraction of that dollar that would have been taxed, while a dollar of tax credits yields a whole dollar of tax benefit. The bisection method also offers improvement near $m=1.33$, as the discontinuity in $\%(m)$ at $1.33$ again prevents IRS convergence nearby.

Some Americans have now used the method suggested by this theorem, altered only by the rounding procedures indicated by the IRS on its forms. Those who have spoken to me about their experiences say that their tax returns have been accepted by the IRS. Some have received three thousand dollars of tax credits for their household, or more, instead of zero dollars.

It should be noted that this procedure also works for general tax returns; at each stage of the bisection, we compute all ``above-the-line" deductions which depend on the health insurance deduction based on the proposed midpoint value for the health insurance deduction. We do this in a way that maximizes their sum, and then we compute the premium tax credit based on modified (adjusted gross) income $M$, the difference between the sum of all sources of income and the sum of the above-the-line deductions.

Finally, there are additional minor details to handle with regard to the implementation of these procedures, including rounding and the proper reconciliation of the premium tax credits taken in advance during the year (to pay health insurance premiums right away) with the actual amount of premium tax credits that should be received (these two numbers may differ due to instability in self-employment income from month to month). However, implementation does not offer any new difficulties, and may be handled as in existing IRS publications such as \cite{IRS974}.

\section{Further Questions}

What are perhaps more mathematically interesting than our ability to give a working algorithm are the tasks of explaining in detail why the IRS procedures break down, and investigating whether more sophisticated algorithms than bisection might be needed if tax laws were changed. For the former, the explanation seems to lie in the discontinuities that the function $C(D)$ possesses in general. For the latter, it should be noted that the bisection procedure, when applied to general tax returns, requires a net monotonicity effect from tax deductions in attempting to maximize the tax benefits
\[
\max_{0\leq D_{i}\leq A_{i}\ :\ D=D_{1}+\cdots+D_{8},\ D_{1}+C(D)\leq Q} C(D)+\overline{t}D
\]
over the ``above-the-line" deductions $D_{1}, \dots, D_{8}$ that get caught in the relevant circular tax relationships, where the $A_{i}$ are tax parameters, $Q$ is the full (unsubsidized) cost of annual health insurance premiums for qualified insurance purchased from the county's government exchange, $C$ is a discontinuous, nonlinear premium tax credit function, and $\overline{t}:=\frac{T(I)-T(I-D)}{I-(I-D)}$ is the effective tax rate for the tax benefit achieved from the total above-the-line deduction $D$ against total taxable income $I$ and household tax function $T$.

As a result, a deduction which ``phases in" with a discontinuous jump upwards may fool the bisection method into finding a point which does not maximize tax benefits; there are currently no such above-the-line deductions of this type, but there were as late as 2005, in the old domestic production tax deduction that might affect, for instance, someone who strikes oil in Texas (the new domestic production tax deduction phases in continuously, and with a weak slope, so it presents no difficulty). So this is a question for further research: what if the law for that tax deduction were changed back to the way it was in 2005? In that case, if the bisection method fails, then we should seek an alternative. A naive binned Newton method (not fully taking into account the constraints) appears to the author to fail, albeit only when $m$ near 1.33 in states, such as Florida and Texas, which have not expanded Medicaid. Recall that, according to the Supreme Court ruling \cite{USC}, each state can choose individually whether or not to expand its Medicaid program in the manner prescribed by the Affordable Care Act. So, a little something more in the realm of discontinuous numerical analysis might be needed for those states.

It has also been asked, by individuals interested in politics, whether one might come up with policy proposals which lead to better mathematical properties for the function $C(d)$. If $C(d)$ were smooth of class $C^{\infty}$, with compact support, then superior computational outcomes could be achieved. Congress may want to achieve additional objectives which are incompatible with this, however, and the possibilities for proposals have not yet been fully explored.


%

\end{document}